\documentclass[journal, twocolumn]{IEEEtran}

\usepackage{caption}

\usepackage{cite}
\usepackage{gensymb}

\usepackage{amsmath,amsthm,amssymb,amsfonts}
\usepackage{algorithmic}
\usepackage{graphicx}
\usepackage{textcomp}
\usepackage{xcolor}
\usepackage{float}
\usepackage{color}

\usepackage{graphicx}
\usepackage[lined,boxed,commentsnumbered, ruled]{algorithm2e}
\usepackage{mathrsfs}
\usepackage{bm}
\usepackage{pgfplots}
\usepackage{tikz}
\usetikzlibrary{arrows}
\usepackage{subfigure}
\usepackage{graphicx,booktabs,multirow}
\usepackage{upgreek}

\definecolor{colorhkust}{RGB}{20,43,140}
\definecolor{colortsinghua}{RGB}{116,52,129}
\definecolor{color1}{RGB}{128,0,0}

\newtheorem{lemma}{Lemma}
\newtheorem{theorem}{Theorem}
\newtheorem{corollary}[theorem]{Corollary}
\newtheorem{proposition}[theorem]{Proposition}
\newtheorem{definition}{Definition}

\newcommand{\rank}{\mathrm{rank}}
\newcommand{\diagg}{\rm{diag}}

\newcommand{\R}{\mathbb{R}}  
\newcommand{\N}{\mathbb{N}}  
\newcommand{\C}{\mathbb{C}} 

\DeclareMathOperator*{\argmin}{argmin}

\mathchardef\re="023C
\mathchardef\im="023D

\begin{document}

\title{Large-Scale Beamforming for Massive MIMO \\
via Randomized Sketching}

%

\author{Hayoung~Choi, \IEEEmembership{Member,~IEEE}, 
Tao~Jiang, \IEEEmembership{Student Member,~IEEE},
Yuanming~Shi~\IEEEmembership{Member,~IEEE},
Xuan~Liu~\IEEEmembership{Member,~IEEE},
                     and Khaled B. Letaief~\IEEEmembership{Fellow,~IEEE} %
\thanks{H. Choi, T. Jiang, and Y. Shi are with School of Information Science and Technology, ShanghaiTech University, Shanghai 201210, China (e-mail:\{hchoi,jiangtao1,shiym\}@shanghaitech.edu.cn).}
\thanks{X. Liu is with the School of Electrical Engineering and Telecom- munications, The University of New South Wales, Australia (e-mail: xuan.liu@unsw.edu.au).}
\thanks{K. B. Letaief is with the Department of Electronic and Computer Engi- neering, The Hong Kong University of Science and Technology, Hong Kong (e-mail: eekhaled@ust.hk).}

\thanks{Part of this work was presented at the IEEE Global Communications Conference, Waikoloa, HI, USA, December 2019 \cite{YM_GC19}.}}


\maketitle

\begin{abstract}
Massive MIMO system  yields significant improvements in spectral and energy efficiency for future wireless communication systems. 
The regularized zero-forcing (RZF) beamforming  is able to provide good performance with the capability of achieving  numerical stability and robustness to the channel uncertainty. 
However, in massive MIMO systems, the matrix inversion operation in RZF beamforming becomes computationally expensive. 
To address this computational issue, we shall propose a novel \emph{randomized sketching} based RZF beamforming approach with low computational complexity. This  is achieved by  solving a linear system via  {randomized sketching} based on  the preconditioned Richard iteration, which guarantees high quality approximations to the optimal solution. We theoretically prove that the sequence of approximations obtained iteratively converges to the exact RZF beamforming matrix linearly fast as the number of iterations increases. Also, it turns out that the system sum-rate for such sequence of approximations converges to the exact one at a linear convergence rate. Our simulation results verify our theoretical findings.
\end{abstract}

\begin{IEEEkeywords}
Regularized zero-forcing beamforming, massive MIMO, randomized sketching algorithm, sketching method.
\end{IEEEkeywords}

\IEEEpeerreviewmaketitle

\section{Introduction}
\IEEEPARstart{W}{ith} the explosive growth in mobile data traffic and number of mobile devices,  as well as the stringent and diverse demands of intelligent mobile services, wireless networks are facing formidable challenges to enable high spectral efficiency and support massive connectivity with low-latency.  To satisfy these requirements, network  densification becomes the key enabling technology. This is achieved by deploying more base stations (BSs) with the storage and computational capabilities, yielding an ultra-dense network (UDN) \cite{shi2017generalized}. In particular, massive multiple-input multiple-output (MIMO)  technique provides an alternative to achieve UDN by simply increasing the number of antennas at the existing BS \cite{rusek2013scaling,DBLP:journals/corr/abs-1902-07678}. The key success is based on the fact that deploying large-scale antenna arrays allows for an exceptional array gain and unprecedented spatial resolution such that the wireless communication system is robust to inter-user interference \cite{lu2014overview}.
Furthermore, the large arrays regime provides the opportunities for asymptotic system analysis, e.g., the high-dimensional random matrix theory can provide deterministic approximations for achievable data rates \cite{hoydis2013massive, Couillent_TIT12}.

Transmit beamforming at the BSs is a key method to optimize the network utility function (e.g., sum-rate) in terms of signal-to-interference-plus-noise ratios (SINRs). However, the resulting beamforming optimization problem is generally very difficult to be solved due to the nonconvexity and high-dimensionality.  
With the known optimal SINRs parameters for maximizing the network utility function, a simple structure for the optimal beamforming can be derived based on the Lagrange duality theory \cite{Bjornson_SPM14}. To find the optimal SINRs parameters, we normally need to solve a sequence of convex subproblems \cite{bjornson2013optimal}. For instance, in the max-min fairness rate optimization problem,
the optimal SINRs   parameters can be found via the bi-section method \cite{zakhour2012base}, wherein a  sequence of convex subproblems are solved.  Although the  general large-scale convex optimization problem can be  solved by the operate splitting method,  it still needs to solve a sequence of subspace projection and cone projection problems in the transformed high-dimensional space for the standard cone program \cite{shi2015large}.  
Instead, the heuristic transmit beamforming, i.e., regularized zero-forcing (RZF) beamforming \cite{peel2005vector} turns out to be the  appealing choice since it performs closely to the optimal beamforming  in terms of sum-rate and has relatively low computational complexity \cite{Bjornson_SPM14}. We thus focus on investigating  RZF beamforming in this paper. 
However, the RZF beamforming needs to compute a matrix inversion with complexity proportional to $MK^2 $, where $ K $ is the number of users served by  $ M $ transmit antennas. This is however computational expensive in massive MIMO scenario where $M \gg K\gg 1$. To tackle this issue, \cite{mueller2016linear,8602444} proposed to replace the matrix inversion in RZF beamformer by a truncated polynomial expansion, but it is not clear that which degree of the polynomial is needed to guarantee the good performance for the system sum-rate.

In recent years,
randomized sketching algorithms\cite{Tropp11,TCS-060,Drineas:2016:RRN:2942427.2842602} have received a great deal of attention 
in order to solve
large-scale matrix computation problems. 
The main idea behind randomized sketching algorithms is to compress a given large-scale matrix to a much smaller matrix by multiplying it by a \emph{random matrix} with certain properties. Very expensive computation can then be operated by the smaller matrix efficiently. In particular, several novel randomized algorithms are proposed for the ridge regression problem \cite{JMLR:v18:17-313,MR3662446, chowdhury2018iterative,pmlr-v80-lin18b}. Inspired by these progresses, we propose a randomized sketching based beamforming method to overcome the computational issues for designing beamformers in massive MIMO systems. Specifically, the randomized sketching RZF beamforming matrix is achieved by solving a linear system by preconditioned Richard iteration \cite{precondition} with the randomized sketching techniques \cite{chowdhury2018iterative}. The proposed randomized sketching RZF beamforming method  has a computational complexity proportional to  $ LK^2 $ with $ L\ll 2M $ as the sketching matrix size. We prove that the beamforming matrix obtained iteratively converges to the RZF beamforming matrix  at a linear convergence rate. Furthermore, we prove that the achievable system sum-rate of the MIMO system with the proposed randomized method converges to the achievable sum-rate given by RZF beamforming linearly as the number of iteration increases. Extensive numerical results are demonstrated to verify our theoretical findings.

\subsection{Outline}
The organization of this paper is as follows. In Sec.~\ref{section2}, the system model and problem statement of estimating the beamforming matrix for a massive MIMO communication system are described.
In Sec.~\ref{randombf}, we propose the randomized sketching method to approximate the beamforming matrix with low-complexity and provide convergence analysis and complexity analysis.
In Sec.~\ref{sumrateprove}, we prove that the system sum-rate of the randomized sketching based beamformer converges to the sum-rate of the RZF beamforming matrix as the number of iterations increases. 
We provide the exact rate of convergence as well.
In Sec.~\ref{sec:simulations}, we numerically evaluate the performance of the randomized sketching based beamforming method. Finally, conclusions are drawn in Sec.~\ref{sec:conclusion}.

\subsection{Notation}
Let $\R$ (resp. $\C$) be the set of real (resp. complex) numbers. 
For a matrix $\bm{A}$,
$\bm{A}_{*i}$ (resp. $\bm{A}_{i*}$) denotes $i$-th column(resp. row) vector of $\bm{A}$.
$\|\bm{A}\|_2$ (resp. $\|\bm A\|_F$) denotes the operator (resp. Frobenius) norm.
For a vector $\bm{x}$, $\| \bm{x} \|_2$ denotes the Euclidean norm.
The superscript $\sf T$ denotes the transpose operator.
$\bm{A}^{\sf H} = \bar{\bm{A}}^{\sf T}$ is a complex conjugate transpose of $\bm{A}$.
The diagonal matrix whose diagonal entries consist of entries of a vector $\bm \lambda$ is denoted by $\diagg \{\bm \lambda\}$.
Denote the identity matrix of size $K$ as $\bm{I}_K$. When the size can be trivially determined by the context, we simply write $\bm{I}$. Denote the zero matrix with size $K \times M$ as $\bm{0}_{K\times M}$. 
Let $\re(\bm{A})$ and $\im(\bm{A})$ denote the real and imaginary parts of a matrix $\bm{A}$, respectively.
For a matrix $\bm{Q} \in \R^{2K\times 2M}$ with $M \geq K$ of rank $2K$, its (thin) Singular Value Decomposition (SVD) 
is the form $\bm{U} \bm{\Sigma} \bm{V}^{\sf T}$ where $\bm{U} \in \R^{2K\times 2K}$ is the matrix of the left singular vectors, $\bm{V} \in \R^{2M\times 2K}$ is the matrix of the right singular vectors, and $\bm{\Sigma} \in \R^{2K\times 2K}$ is a diagonal matrix whose diagonal entries are the singular value of $\bm{Q}$. We denote the singular values of a matrix as $\sigma_i$.
We denote the matrix of the top $j$ left singular vectors as $\bm{U}_j \in  \R^{2K\times j}$ and
the matrix of the bottom $2K-j$ left singular vectors as $\bm{U}_{j,\perp} \in  \R^{2K\times (2K-j)}$.

\section{System Model and Problem Statement}
\label{section2}

\subsection{System Model}
We consider  a  single-cell massive MIMO  system consisting  of one BS equipped with $ M $ antennas  and $ K $ single-antenna users, where $ M \geq K $.  During the downlink
transmission, the received signal at the $ k $-th user is given by
\setlength\arraycolsep{2pt}
\begin{equation}\label{eq:sysmod1}
y_k = \bm h_k^{\sf H}\left(\sum_{i=1}^K\bm{w}_is_i\right)+n_k,\quad k=1,\ldots,K,
\end{equation}
where $\bm{w}_i\in\mathbb{C}^M$ is the transmit beamforming vector from the BS for data symbol $s_i$ to user $i$, $ \bm h_k\in\mathbb{C}^M $  is the channel propagation coefficients from the  BS to the $ k $-th user, and $ n_k\sim\mathcal{CN}(0,\sigma^2)$ is the additive noise (i.e., $n_k$ is
a circularly symmetric complex Gaussian random distribution with  mean $0$ and variance $\sigma^2$). Therefore, the
SINR at the $ k $-th user is given as
\begin{equation}\label{eq:SINR}
{\sf{SINR}}_k (\bm{W}) := \frac{|\bm h_k^{\sf H}\bm w_k|^2}{\sum_{j\ne k}|\bm h_k^{\sf
H}\bm w_j|^2+\sigma^2},
\end{equation}
where $ \bm W = [\bm{w}_1, \cdots, \bm{w}_K]\in\mathbb{C}^{M\times K}$ is
the aggregative beamforming matrix with the total transmit power limited by $P>0$, i.e.,
 
\begin{equation}\label{eq:power}
\|\bm W\|_F^2=\sum_{k=1}^{K} \|\bm w_k\|_2^2 \le P .
\end{equation}
The achievable system  sum-rate $ R(\bm W) $ is thus given by 
\begin{equation}\label{eq:R}
R(\bm W):= \sum_{k=1}^{K}\log(1+{\sf{SINR}}_k (\bm{W})).
\end{equation}
One of  the main goal of  transmit beamforming is to maximize the achievable system sum-rate. However, it is generally   computationally demanding to find the optimal beamforming matrix $ \bm W $ \cite{bjornson2013optimal}.  

\subsection{Regularized Zero-Forcing Beamforming}


Although there are various precoding techniques such as matched filter, zero forcing, regularized zero-forcing, truncated polynomial expansion, and phased zero forcing \cite{8169014}, this article considers the suboptimal  beamforming approach, \emph{regularized zero-forcing (RZF) beamforming} \cite{peel2005vector}, which is known to have the capability achieving robustness and numerical stability to the channel uncertainty \cite{peel2005vector, Bjornson_SPM14}. 
RZF precoder has been considered as the state-of-the-art linear precoder for MIMO wireless communication systems. 
Since we focus on the computational issues of RZF, we consider the following RZF with equal power allocation for simplification \cite{hoydis2013massive}
\begin{eqnarray}\label{eq:bfd}
\label{rzfbf}
\bm{W}^{\ast}&=&\beta\left(\bm{I}_M+{\gamma\over{\sigma^2}}\bm{H}^{\sf{H}}\bm{H}\right)^{-1}\bm{H}^{\sf{H}}\nonumber\\
&=&\beta\bm{H}^{\sf{H}} \left(\bm{I}_K+{\gamma\over{\sigma^2}}\bm{H}\bm{H}^{\sf{H}}\right)^{-1},
\end{eqnarray} 
where
$ \bm{H} = [\bm{h}_1, \cdots, \bm{h}_K]^{\sf H} \in\mathbb{C}^{K\times M}$ is the channel matrix, 
 $ \gamma>0 $ is an optimal regularizer, and $ \beta >0$ is a normalization parameter to satisfy the power constraint  \eqref{eq:power}. In particular, $\gamma$ can be derived as $ \gamma = P/K $ in the symmetric scenario, where the channels are equally strong\cite{peel2005vector}.
 
 \subsection{Complexity Issues in Massive MIMO}
 The main computational complexity for computing \eqref{eq:bfd} lies in computing the matrix inversion directly, which 
leads  $\mathcal{O}(MK^2+K^3)$ computational complexity. To support ultra-low latency communications in massive MIMO systems, it becomes critical to design large-scale precoding algorithm with low computation complexity \cite{mueller2016linear, 8602444}. As fast inversions of large-scale matrices in every coherence period needs to be performed, it is desired to find efficient algorithms to 
reduce the high
computational complexity with  performance guarantees.


In this paper, we shall develop the randomized sketching based precoding algorithm 
to compute the large-scale RZF beamforming matrix $\bm{W}^{\ast}$ in \eqref{rzfbf}. This is based on the key observation that the large-scale array regime, i.e., $M\gg K$, offers the opportunity for dimension reduction in \eqref{rzfbf}, thereby reducing the computational complexity while guaranteeing the high performance accuracy. Specifically, we develop the scalable algorithm for computing $\bm{W}^{\ast}$ in \eqref{rzfbf} based on the principles of \emph{Randomized Numerical Linear Algebra} \cite{Drineas:2016:RRN:2942427.2842602}. In particular, the theoretical guarantees for the achievable system sum-rate \eqref{eq:R} using the randomized sketching based beamforming method will be presented in Sec.~\ref{sumrateprove}.

\section{Randomized Sketching for Large-Scale Beamforming}
\label{randombf}

\subsection{Randomized Sketching Algorithm}
Randomized sketching algorithm exploits \emph{randomization} as a computational resource to develop improved algorithms for large-scale matrix computation problems. The key idea of randomized algorithm is to compress a given large-scale matrix to a much smaller matrix by multiplying it by a \emph{random matrix} with certain properties. Very expensive computation can then be performed on the smaller matrix efficiently. 
For a given matrix $\bm{A}$ and a random matrix $\bm{S}$, the technique of replacing $\bm{A}$ by $\bm{S} \bm{A}$ is known as a \emph{sketching technique} and $\bm{S} \bm{A}$ is referred to as a \emph{sketch} of $\bm{A}$. 
Such $\bm{S}$ is called a \emph{sketching matrix}.

Sketching technique can be accomplished by \emph{random sampling} or \emph{random projection}.
For random sampling method, the sketch consists of a small number of carefully-sampled and rescaled columns/rows of matrix $\bm{A}$.
On the other hand, for random projection method, the sketch consists of a small number of linear combinations of the columns/rows of $\bm{A}$.
We will discuss various construction for the random matrix $\bm{S}$
in Section \ref{subsec:sketching}.

Sketching technique has been extensively studied for a decade \cite{Tropp11,TCS-060,Drineas:2016:RRN:2942427.2842602}.
Recently, the widespread use of sketching as a tool for matrix computations yields many novel results in many fields, especially in machine learning \cite{JMLR:v18:17-313,MR3662446,Bhojanapalli2014CoherentMC,Bi:2013:EMC:3042817.3042982}.


\subsection{Randomized Sketching Based RZF Beamforming}
The first key observation is that \eqref{eq:bfd} can be expressed as the matrix ridge regression problem as follows \cite{chowdhury2018iterative}:
\begin{equation}\label{eq:rg}
\bm{W}^{\ast} = \underset{\bm W\in\mathbb{C}^{M\times K}}{\argmin}~\|\bm H\bm W- \lambda\beta\bm I_K\|_F^2+\lambda\|\bm W\|_F^2,
\end{equation} 
where $\lambda=\frac{\sigma^2}{\gamma}$. 
To facilitate algorithm design in real field, we  focus on solving the equivalent real counterpart of \eqref{eq:rg}:
 \begin{equation}\label{eq:real_main_formulation111}
\bm{M}^{\ast} = \underset{\bm M \in\R^{2M\times K} }{\argmin} ~ \|\bm{Q} \bm{M}-  \bm \Lambda\|_F^2+\lambda\|\bm{M} \|_F^2,
 \end{equation} 
where
\begin{align*}
\bm M=
\begin{bmatrix}
\re{(\bm{W})} \\
\im{ (\bm{W} )}
\end{bmatrix} ,
 \bm Q=
\begin{bmatrix}
\re{ (\bm{H}) } & - \im{ (\bm{H}) } \\
\im{ (\bm{H}) } & ~\re{ (\bm{H}) }
\end{bmatrix}, 
\bm \Lambda =
\begin{bmatrix}
\re{ (\lambda \beta \bm{I}_K)} \\
\im{  (\lambda\beta \bm{I}_K)}
\end{bmatrix} .
\end{align*} 
Note that since $\lambda, \beta>0$, $\im{  (\lambda\beta \bm{I}_K)} = \bm{0}$.
Then the optimal solution of (\ref{eq:real_main_formulation111}) takes the form,
\begin{equation}\label{eq:rzf_sol}
\bm M^{\ast} = \bm{Q}^{\sf{T}} (\bm{Q} \bm{Q}^{\sf{T}}+\lambda\bm I_{2K})^{-1}  \bm{\Lambda}.
\end{equation} 
Given the matrix $ \bm M^{\ast} $, it is trivial to obtain the complex RZF beamforming matrix $ \bm W^{\ast} $ in \eqref{eq:bfd}.

%

Iterative methods provide the solution to the linear system $\bm{A}\bm{x} = \bm{b}$ as the limit of a sequence $\bm{x}^{(j)}$, and usually involve matrix $\bm{A}$ only through multiplications by given vectors. 
Generally, any iterative method is based on a suitable splitting of the matrix $\bm{A}$ with $\bm{A} = \bm{E} - \bm{N}$, where $\bm{E}$ is nonsingular.
Then the sequence $\{ \bm{x}^{(j)} \}$ is generated as follows:
\begin{equation}
\bm{E} \bm{x}^{(j+1)} = \bm{N} \bm{x}^{(j)} + \bm{b} \quad \text{for all } j\in \N,
\end{equation}
where $\bm{x}^{(0)}$ is a given initial vector.
Equivalently, such iteration can be restated as 
$\bm{x}^{(j+1)} = \bm{x}^{(j)} + \bm{E}^{-1}\bm{r}^{(j)}  ~\text{for all } j\in \N$,
where $\bm{r}^{(j)}:= \bm{b} - \bm{A} \bm{x}^{(j)}$ is the residual at the step $j$, where $\bm{E}$ is called \emph{preconditioner} for $\bm{A}$.
The following iteration is called \emph{preconditioned Richardson iteration}\cite{precondition}:
\begin{equation}\label{eq:preconditioned_Richardson}
\bm{x}^{(j+1)} = \bm{x}^{(j)} + \alpha_j \bm{E}^{-1}\bm{r}^{(j)} \quad \text{for all } j\in \N,
\end{equation}
where $\alpha_j \neq 0$ is the real acceleration parameter.

We present the novel sketching based randomized beamforming in Algorithm  \ref{main:Algorithm}, which  iteratively computes a sequence of matrixes $\widetilde{\bm{M}}^{(j)} \in \mathbb{R}^{2M\times K}$ for $j= 1,\dots,t$ and returns the approximation $\widehat{\bm{M}}^{(t)} = \sum_{j=1}^t \widetilde{\bm{M}}^{(j)}$ to the true solution matrix of \eqref{eq:rzf_sol}.
In fact, it can be viewed as a preconditioned Richardson iteration.
Indeed, for a given $\bm{Y}^{(j)}$ in Algorithm \ref{main:Algorithm},
we denote $\widehat{\bm{Y}}^{(t)} = \sum_{j=1}^t \bm{Y}^{(j)}  $. Note that our solution is $\bm{M}^{(t)} = \bm{Q}^{\sf T} \widehat{\bm{Y}}^{(t)}$.
By (i) and (iii) in Algorithm \ref{main:Algorithm}, we have  
\begin{equation}\label{eq:recurrence}
\bm{\Lambda}^{(j)}
=  \bm{\Lambda}^{(j-1)}  - ( \bm{Q} \bm{Q}^{\sf T} + \lambda \bm{I} )  \bm{Y}^{(j-1)}. 
\end{equation}
Applying the recurrence relation \eqref{eq:recurrence} successively, it follows that 
\begin{align*}
\bm{\Lambda}^{(j)}
&= \bm{\Lambda}^{(j-2)}  - ( \bm{Q} \bm{Q}^{\sf T} + \lambda \bm{I} )  \bm{Y}^{(j-2)} -  ( \bm{Q} \bm{Q}^{\sf T} + \lambda \bm{I} ) \bm{Y}^{(j-1)}\\
&= \bm{\Lambda}^{(j-2)}  - ( \bm{Q} \bm{Q}^{\sf T} + \gamma \bm{I} )  (\bm{Y}^{(j-2)}  + \bm{Y}^{(j-1)})\\
& \vdots\\
& = \bm{\Lambda}^{(1)}  - ( \bm{Q} \bm{Q}^{\sf T} + \lambda \bm{I} )  (\bm{Y}^{(j-2)}  + \cdots +\bm{Y}^{(1)})\\
&= \bm{\Lambda}  - ( \bm{Q} \bm{Q}^{\sf T} + \lambda \bm{I} )  \widehat{\bm{Y}}^{(j-1)} .
\end{align*}
Then it holds that 
\begin{align*}
& \widehat{\bm{Y}}^{(t)}  =  \widehat{\bm{Y}}^{(t-1)} + \bm{Y}^{(t)} \\
&= \widehat{\bm{Y}}^{(t-1)} +(\bm{Q} \bm{S} \bm{S}^{\sf T} \bm{Q}^{\sf T} + \lambda \bm{I})^{-1} \bm{\Lambda}^{(t)}\\
& =  \widehat{\bm{Y}}^{(t-1)} +  (\bm{Q} \bm{S} \bm{S}^{\sf T} \bm{Q}^{\sf T} + \lambda \bm{I})^{-1}  (\bm{\Lambda}  - ( \bm{Q} \bm{Q}^{\sf T} + \lambda \bm{I} )  \widehat{\bm{Y}}^{(t-1)} ).
\end{align*}
Thus, Algorithm \ref{main:Algorithm} can be formulated as a preconditioned Richard iteration to solve the linear system
\begin{equation}
 (\bm{Q}\bm{Q}^{\sf T} + \lambda \bm{I}_{2K}) \bm{Y} = \bm{\Lambda},
\end{equation}
with preconditioner $\bm{E} =  (\bm{Q} \bm{S} \bm{S}^{\sf T} \bm{Q}^{\sf T} + \lambda \bm{I}_{2K})$ and $\alpha_j =1$ for all $j$ in  \eqref{eq:preconditioned_Richardson}.

%
%

\begin{algorithm}[t!]
\KwIn{$\bm{Q}\in \R^{2K \times 2M}$, $\bm{\Lambda} \in \R^{2K \times K }$, $\lambda >0$; number of iterations $t>0$;
sketching matrix $\bm{S} \in \mathbb{R}^{2M\times L}$; \\
{\bf 
Initialize: $\bm{\Lambda}^{(0)} \leftarrow \bm{\Lambda}$, $\widetilde{\bm{M}}^{(0)}  \leftarrow \bm{0}_{2M\times K}$, $\bm{Y}   \leftarrow  \bm{0}_{2K\times K}$};\\
{\bf for $j=1$ to $t$ do }\\
\begin{itemize}
\item[(i)] {\bf $\bm{\Lambda}^{(j)}  \leftarrow \bm{\Lambda}^{(j-1)} - \lambda \bm{Y}^{(j-1)} - \bm{Q}~ \widetilde{\bm{M}}^{(j-1)} $};

\item[(ii)] {\bf $\bm{Y}^{(j)}  \leftarrow  (\bm{Q} \bm{S} \bm{S}^{\sf T} \bm{Q}^{\sf T} + \lambda \bm{I}_{2K})^{-1} \bm{\Lambda}^{(j)}$};

\item[(iii)]
{\bf $\widetilde{\bm{M}}^{(j)}   \leftarrow \bm{Q}^{\sf T} \bm{Y}^{(j)}$};\\

\end{itemize}
{\bf end for}
}

\KwOut{Approximate solution matrix $\widehat{\bm{M}}^{(t)} = \sum_{j=1}^t  \widetilde{\bm{M}}^{(j)} $.
}
    \caption{Randomized Sketching Based Beamformer}\label{main:Algorithm}
\end{algorithm}


Algorithm~\ref{main:Algorithm} iteratively computes a sequence of matrices  
$\widetilde{\bm{M}}^{(j)}$ for $j=1,\ldots, t$ and returns the approximation
$\widehat{\bm{M}}^{(t)} = \sum_{j=1}^t  \widetilde{\bm{M}}^{(j)} $
to the true solution $\bm{M}^{\ast}$ in  \eqref{eq:rzf_sol}. 
Equivalently, it computes the approximation
$\widehat{\bm{W}}^{(t)} = \sum_{j=1}^t  \widetilde{\bm{W}}^{(j)} $
to the true solution $\bm{W}^{\ast}$ in  \eqref{eq:bfd}.
We call such approximation $\widehat{\bm{W}}^{(t)}$ \emph{a randomized sketching based beamformer}.

Algorithm \ref{main:Algorithm} uses the sketching matrix for the preconditioner in order to improve the rate of convergence and reduce the computational complexity.
Specifically, using the sketching matrix $\bm{S}\in \mathbb{R}^{2M\times L}$ with $L\ll 2M$, the preconditioner  $\bm{E} =  (\bm{Q} \bm{S} \bm{S}^{\sf T} \bm{Q}^{\sf T} + \lambda \bm{I}_{2K})$ can be computed by  matrix $\bm{Q} \bm{S}$ with much smaller size.

\subsection{Convergence Analysis}

The convergence analysis depends on the selected sketching matrix, which satisfies the constraint \eqref{eq:contraint1}.
Theorem {\ref{thm:main1}} presents a quality-of-approximation result under
the assumption that the sketching matrix  
satisfies the constraint \eqref{eq:contraint1}.

\begin{theorem}\label{thm:main1}
Assume that for some constant $0 < \varepsilon <1$, the sketching matrix $\bm{S} \in \R^{2M\times L}$ satisfies the following constraint
\begin{equation}\label{eq:contraint1}
\| \bm{V}^{\sf T} \bm{S} \bm{S}^{\sf T} \bm{V} - \bm{I}_{2K} \|_2 \leq \frac{\varepsilon}{2}, 
\end{equation}
where $\bm{V} \in \mathbb{R}^{2M\times 2K}$ is the matrix of right singular vectors of $ \bm{Q}$.
Then, after $t$ number of iterations, the approximation $\widehat{\bm{W}}^{(t)}$ returned by Algorithm \ref{main:Algorithm} satisfies
\begin{equation*}
\|\widehat{\bm{W}}^{(t)} - \bm{W}^*\|_F \leq \varepsilon^t \|\bm{W}^* \|_F,
\end{equation*}
where $\bm{W}^*$ is the true value of the RZF beamforming matrix in \eqref{eq:bfd} in the complex version.
\end{theorem}

\begin{proof}
Note that by \cite{Saunders98ridgeregression}, \eqref{eq:rzf_sol} can be also expressed as
\begin{equation}\label{eq:ridge2}
\bm{M}^* = (\bm{Q}^{\sf T} \bm{Q} + \lambda \bm{I}_{2M} )^{-1} \bm{Q}^{\sf T} \bm{\Lambda}.
\end{equation}

Then each column of $\bm{M}^*$ can be considered as the solution of the following optimization problem
\begin{equation}\label{eq:columnM}
 \underset{  \bm{M}_{*i}  \in\mathbb{R}^{2M} }{\arg \min} ~ \| \bm{Q}~\bm{M}_{*i} -\bm{\Lambda}_{*i}  \|_2^2+\lambda\| \bm{M}_{*i} \|_2^2,
\end{equation}
for each $i=1,\ldots, K$.
Recall that $\bm{M}_{i*}$ and $\bm{\Lambda}_{i*}$ is the $i$-th column of $\bm{M}$ and $\bm{\Lambda}$, respectively.
By Theorem 1 in \cite{chowdhury2018iterative}, it follows that 
\begin{equation*}
\|\widehat{\bm{M}}_{*i}^{(t)} - (\bm{M}^*)_{*i} \|_2 \leq \varepsilon^t \| (\bm{M}^*)_{*i} \|_2 
\end{equation*}
for all $i=1,\ldots,K$.
Then we have
\begin{align*}
\| \widehat{\bm{M}}^{(t)} - \bm{M}^* \|_F^2 
& = \sum_{i=1}^K \| \widehat{\bm{M}}_{*i}^{(t)} - (\bm{M}^*)_{*i} \|_2^2\\
& \leq \varepsilon^{2t} \sum_{i=1}^K \|  (\bm{M}^*)_{*i} \|_2^2\\
& \leq \varepsilon^{2t} \| \bm{M}^* \|_{F}^2.
\end{align*}
Clearly,
$\|\widehat{\bm{W}}^{(t)} - \bm{W}^*\|_F = \| \widehat{\bm{M}}^{(t)} - \bm{M}^* \|_F$
and $\| \bm{W}^* \|_{F} = \| \bm{M}^* \|_{F}$.
\end{proof}

To check whether a sketching matrix $\bm{S}$
satisfies  \eqref{eq:contraint1},
a number of columns $L$ that is proportional to $2K\log{(2K)}$ is required (see Theorem \ref{thm:number_samples}).
Thus, the running time of any algorithm that computes the sketch $\bm{Q} \bm{S}$ is also proportional to $2K\log{(2K)}$.
To reduce the running time, it would be much better to use a parameter which is significantly smaller than $2K$. For simplicity of exposition, we will assume that the rank of $\bm{Q}$ is $2K$.

In the context of ridge regression, a much more important quantity than the rank 
of $\bm{Q}$ is the {\emph{(effective) degrees of freedom}} of $\bm{Q}$ as follows \cite{Dijkstra2014}:
\begin{equation}
d_{\lambda} = \sum_{i=1}^{2K} \frac{\sigma_i^2}{\sigma_i^2+{\lambda}},
\end{equation}
where $\sigma_i$ are the singular values of $\bm{Q}$ and $\lambda=\frac{\sigma^2}{\gamma}$.
Since $\lambda>0 $, it is trivial that $d_{\lambda} \leq 2K$.
That is, the degrees of freedom $d_{\lambda}$ is upper bounded by the rank of $\bm{Q}$.

Define a diagonal matrix $\bm{\Sigma}_\lambda \in \mathbb{R}^{2K\times 2K}$ whose $i$-th diagonal entry is given by
\begin{equation}\label{eq:sum}
(\bm{\Sigma}_{\lambda})_{ii}=\sqrt{\frac{\sigma_i^2}{\sigma_i^2+\lambda}},\quad i=1,\dots,2K,
\end{equation}
where 
$\sigma_i$ is the $i$-th singular value of $\bm{Q}$ and $\lambda=\frac{\sigma^2}{\gamma}$.



Now we provide a weaker constraint with the effective degrees of freedom.
\begin{theorem}\label{thm:main2}
Assume that for some constant $0 < \varepsilon <1$, the sketching matrix $\bm{S} \in \R^{2M\times L}$ satisfies the following constraint
\begin{equation}\label{eq:sketching_matrix4}
\| \bm{\Sigma}_\lambda \bm{V}^{\sf T} \bm{S} \bm{S}^{\sf T}  \bm{V\Sigma}_\lambda - \bm{\Sigma}_{\lambda}^2 \|_2 \leq \frac{\varepsilon}{4\sqrt{2}},
\end{equation}
where $\bm{V} \in \mathbb{R}^{2M\times 2K}$ is the matrix of right singular vectors of $ \bm{Q}$.
Then, after $t$ number of iterations, the approximation
$\bm{\widehat{W}}^{(t)}$ returned by Algorithm~\ref{main:Algorithm} satisfies
\begin{equation}\label{eq:th2}
\|\widehat{\bm{W}}^{(t)} - \bm{W}^*\|_F \leq \frac{\varepsilon^t}{\sqrt{2}} \Big(\|\bm{W}^* \|_F^2 +\frac{1}{2\lambda} \| \bm{U}^{\sf T}_{\xi,\perp}\bm{\Lambda} \|_F^2 \Big)^{\frac{1}{2}} ,
\end{equation}
where $\xi$ is an integer number such that 
$\sigma_{\xi+1}^2 \leq \lambda \leq \sigma_{\xi}^2$,
 $\bm{U}_{j,\perp} \in  \R^{2K\times (2K-j)}$ is the matrix of the bottom $2K-j$ left singular vectors of the matrix $\bm{Q}$,
and $\bm{W}^*$ is the true value of the RZF beamforming matrix in \eqref{eq:bfd} in the complex version.
\end{theorem}
\begin{proof}
%
%
Since each column of $\bm{M}^*$ can be considered as the solution of \eqref{eq:columnM},
by Theorem 2 in \cite{chowdhury2018iterative}, it follows that 
\begin{equation*}
\|\widehat{\bm{M}}_{*i}^{(t)} - \bm{M}_{*i}^* \|_2 \leq \frac{\varepsilon^t}{2} \Big(\|\bm{M}_{*i}^* \|_2 +\frac{1}{\sqrt{2\lambda}} \| \bm{U}^{\sf T}_{\xi,\perp}\bm{\Lambda}_{*i} \|_2 \Big),
\end{equation*}
for all $i=1,\ldots,K$.
Then we have
\begin{align*}
\| \widehat{\bm{M}}^{(t)} - \bm{M}^* \|_F^2 
& = \sum_{i=1}^K \| \widehat{\bm{M}}_{*i}^{(t)} - \bm{M}_{*i}^* \|_2^2\\
& \leq 
\sum_{i=1}^K  \frac{\varepsilon^{2t}}{4} \Big(\|\bm{M}_{*i}^* \|_2 +\frac{1}{\sqrt{2\lambda}} \| \bm{U}^{\sf T}_{\xi,\perp}\bm{\Lambda}_{*i} \|_2 \Big)^2\\
& \leq 
\sum_{i=1}^K  \frac{\varepsilon^{2t}}{2} \Big(\|\bm{M}_{*i}^* \|_2^2 +\frac{1}{2\lambda} \| \bm{U}^{\sf T}_{\xi,\perp}\bm{\Lambda}_{*i} \|_2^2 \Big)\\
& \leq 
  \frac{\varepsilon^{2t}}{2} \Big(\|\bm{M}^* \|_F^2 +\frac{1}{2\lambda} \| \bm{U}^{\sf T}_{\xi,\perp}\bm{\Lambda} \|_F^2 \Big).
\end{align*}
\end{proof}

This improved dependency on $d_\lambda$ instead of the rank of matrix $\bm{Q}$ results in a mild loss in accuracy.
$\lambda$ can be thought of as regularizing the bottom $2K-\xi$ singular values of the matrix $\bm{Q}$, since it dominates them.
Theorem~\ref{thm:main2} presents a quality-of-approximation result, which uses a relative-additive error approximation.
The term $\| \bm{U}^{\sf T}_{\xi,\perp}\bm{\Lambda} \|_F$ is a norm of the part of matrix $\bm{\Lambda}$ that lies on the regularized component of $\bm{Q}$. As the increase of this part, the
quality of the approximation will become worsen. The error decreases exponentially fast with the number of iterations.

The bounds of \eqref{eq:contraint1} and \eqref{eq:sketching_matrix4} guarantee high-quality approximations to the optimal solution. Constraint \eqref{eq:contraint1} can be satisfied by constructing the sampling-and-rescaling matrix $\bm{S}$ whose size depends on the rank of matrix $\bm{Q}$, and Theorem \ref{thm:main1} guarantees relative error approximations. The second constraint \eqref{eq:sketching_matrix4} can be satisfied by sampling with respect to the ridge leverage scores, which construct the sampling-and-rescaling matrix $\bm{S}$ whose size depends on the degrees of freedom $d_\lambda$, and Theorem~\ref{thm:main2} guarantees relative error approximations.

\subsection{Sketching Matrices}\label{subsec:sketching}
Matrix sketching attempts to reduce the size of large matrices while minimizing the loss of spectral information that is useful in tasks like linear regression.
Matrix sketching algorithms use a typically randomized procedure to compress $\bm{Q}\in \R^{2K\times 2M}$ into an approximation (or ``sketch") $\bm{C} \in \R^{2K \times L}$ with many fewer columns $(L \ll 2M)$.
Matrix sketching can be accomplished by \emph{random sampling} or \emph{random projection}.
Random projection algorithms construct $\bm{C}$ by forming $L$ random linear combinations of the columns in $\bm{Q}$. On the other hand, random sampling algorithms construct $\bm{C}$ by selecting and possibly rescaling a $L$ columns in $\bm{Q}$.
In the latter case, we call a sketching matrix $\bm{S}$ as the \emph{sampling-and-rescaling matrix}. 

Sampling itself is simple and extremely efficient. A simple way to perform this random sampling would be to select those columns uniformly at random in i.i.d. trials, which mean 
$p_1=p_2 = \cdots = p_{2M} = \frac{1}{2M}$. 
A more sophisticated and much more powerful way to do this would be to construct an important sampling routines which select columns using carefully chosen, non-uniform probabilities $\{p_i\}_{i=1}^n$.
It is known that variations on the standard \textquotedblleft statistical leverage scores" give probabilities that are provably sufficient for approximations such as low-rank approximation.
Many of these probabilities are modifications on the standard statistical leverage scores.
\begin{definition}
The \emph{(statistical) leverage score} of the $i^{th}$ column $\bm{Q}_{*i}$ of $\bm{Q}$ is defined as:
\begin{equation}
     \tau_i = \bm{Q}_{*i}^{\sf T}(\bm{QQ}^{\sf T})^{\dagger}\bm{Q}_{*i},
\end{equation}
for $i=1,2,\ldots, 2M$.
\end{definition}
Here, $\dagger$ denotes the Moore-Penrose pseudoinverse of a matrix. When $\bm{Q}\bm{Q}^{\sf T}$ is full rank,
 $(\bm{Q}\bm{Q}^{\sf T})^\dagger = (\bm{Q}\bm{Q}^{\sf T})^{-1}$. $\tau_i$ measures how important $\bm{Q}_{*i}$ is in composing the range of $\bm{Q}$. It is maximized at 1 when $\bm{Q}_{*i}$ is linearly independent from $\bm{Q}$'s other columns and decreases when many other columns approximately align with $\bm{Q}_{*i}$ or when $\|\bm{Q}_{*i} \|_2$ is small.

\emph{Leverage score sampling} sets $p_i$ proportional to the
(exact or approximate) leverage scores $\tau_i$ of $\bm{Q}$.
The leverage scores are used in fast sketching algorithms for linear regression and matrix preconditioning\cite{Drineas:2006:SAL:1109557.1109682,Cohen:2015:USM:2688073.2688113,4031351}.

Notably, leverage scores are defined in terms of $\bm{Q}_{*i}$, which is not always unique and
regardless can be sensitive to matrix perturbations. As a result, the scores can change drastically
when $\bm{Q}$ is modified slightly or when only partial information about the matrix is known. This largely limits the possibility of quickly approximating the scores with sampling algorithms, and motivates our adoption of a new leverage score.
Rather than using leverage scores based on $\bm{Q}_{*i}$, we employ regularized scores called \emph{ridge leverage scores}, which have been used for approximate kernel ridge regression \cite{Alaoui:2015:FRK:2969239.2969326} and in works on iteratively computing standard leverage scores \cite{DBLP:journals/corr/KapralovLMMS14,DBLP:journals/corr/KapralovLMMS14}. 
For a given regularization parameter $\lambda$, we define the $\lambda$-ridge leverage score as:
\begin{equation}
     \tau_i^{\lambda} = \bm{Q}_{*i}^{ \sf T}(\bm{QQ}^{\sf T} + \lambda \bm{I}_{2K})^{-1}\bm{Q}_{*i}.
\end{equation}

Let $\bm{Q}_{\ell}$ be the best low-rank approximation for $\bm{Q}$ with respect to the Frobenius norm.
In other words, 
$$
\bm{Q}_{\ell} = \underset{\bm{X}:\rank(\bm{X})\leq \ell }{\arg\min} \| \bm{Q} - \bm{X} \|_F.
$$
Note that $\bm{Q}_\ell$ can be expressed as  $\bm{U}_\ell \bm{U}_\ell^{\sf T} \bm{Q}$.
That is, the best rank $\ell$ approximation can be found by projecting $\bm{Q}$ onto the span of its top $\ell$ singular vectors.
We will always set $\lambda =\|\bm{Q} -\bm{Q}_\ell \|^2_F/\ell$ as follows.


\begin{definition}
The \emph{ridge leverage score} of the $i^{th}$ column $\bm{Q}_{*i}$ of $\bm{Q}$ with respect to the ridge parameter $\lambda>0$ is defined as:
\begin{equation}
     \bar{\tau}_i = \bm{Q}_{*i}^{ \sf T} \Big(\bm{QQ}^{\sf T} + \frac{\|\bm{Q} -\bm{Q}_\ell \|^2_F}{\ell} \bm{I}_{2K} \Big)^{-1}\bm{Q}_{*i},
\end{equation}
for $i=1,2,\ldots, 2M$.
\end{definition}
Note that the ridge leverage score can also be expressed as 
\begin{equation*}
  \bar{\tau}_i = \| (\bm{V} \bm{\Sigma}_{\lambda})_{i*} \|_2^2  \quad \text{for all } i=1,2,\ldots, 2M,
\end{equation*}
where $\bm{V} \in \mathbb{R}^{2M\times 2K}$ is the matrix of right singular vectors of $ \bm{Q}$ and $\bm{\Sigma}_{\lambda}$ is defined as \eqref{eq:sum}.
The constraint  \eqref{eq:sketching_matrix4} can also be satisfied by sampling with respect to the ridge leverage scores
\cite{Alaoui:2015:FRK:2969239.2969326}. 
The difference is that, instead of having the column size $L$ of the matrix $\bm{S}$ depend on $2K$, it now depends on $d_\lambda$, which could be considerably smaller. 
Indeed, it follows that by sampling-and-rescaling $\mathcal{O}( d_\lambda \ln d_\lambda)$ from the design matrix $\bm{Q}$ (using either exact or approximate ridge leverage scores).

\begin{algorithm}[t!]
\KwIn{Sampling probabilities $p_i$, $i=1,\ldots,2M$; integer $L \ll 2M$;\\
$\bm{S} \leftarrow \bm{O}_{2M\times L}$ ;\\
{\bf for $j=1$ to $L$ do}\\
{\bf Pick $i_j \in \{  1,\ldots, 2M  \} $ with $\mathbb{P}(i_j =i) = p_i$};\\
{\bf $\bm{S}_{i_j,j} \leftarrow (Lp_{i_j})^{-\frac{1}{2}}$};  \\
{\bf end for }
} 

\KwOut{Sampling-and-rescaling matrix $\bm{S}$;}
    \caption{ Construct sampling-and-rescaling matrix \label{Algorithm2}}
\end{algorithm}

In this article we only consider the \emph{sampling-and-rescaling matrix} for a sketching matrix $\bm{S}$. Algorithm \ref{Algorithm2} provides the construction of it.
The following theorems show how many sampled columns guarantee that the sketching matrix holds the constraint \eqref{eq:contraint1}.
This theorem is adopted from Theorem 3 in \cite{chowdhury2018iterative}, so the proof is omitted.
\begin{theorem}\label{thm:number_samples}
Let $\bm{V}\in \mathbb{R}^{2M\times 2K}$ be the matrix of right singular vectors of $\bm{Q}$.
Let $\bm{S}$ be constructed by Algorithm \ref{Algorithm2} with the sampling probabilities $p_i = \| \bm{V}_{i*} \|_2^2 /2K$ for $i=1,\ldots, 2M$. Let $\delta$ be a failure probability and let $0< \varepsilon \leq 1$ be an accuracy parameter. If the number of sampled columns $L$ satisfies
\begin{equation}\label{eq:set_equal}
L \geq \frac{16K}{3 \varepsilon^2} \log{ \bigg( \frac{4(1+ 2K)}{\delta}    \bigg)},
\end{equation} 
then, with probability at least $1-\delta$,
\begin{equation}\label{eq:sketching_matrix2}
\| \bm{V}^{\sf T} \bm{S} \bm{S}^{\sf T}  \bm{V} - \bm{I}_{2K} \|_2 \leq \varepsilon.
\end{equation}
\end{theorem}

The sampling probabilities $p_i = \| \bm{V}_{i*} \|_2^2 /2K$ are the column \emph{leverage scores} \cite{chowdhury2018iterative} of the channel matrix $\bm{Q}$. 
Setting $L=\mathcal{O} (\varepsilon^{-2} K \ln{K})$ suffices to satisfy the condition \eqref{eq:contraint1}.
\cite{cohen_et_al:LIPIcs:2016:6278} demonstrated
a construction for such $\bm{S}$ with $L=\mathcal{O} (\varepsilon^{-2}K)$ columns such that, for 
$\bm{Q} \in \R^{2M \times 2K}$, the product $\bm{Q} \bm{S}$ can be computed in time $\mathcal{O}(nnz(\bm{Q})) + \mathcal{O}(K^3/ \varepsilon^\lambda )$ for some constant $\lambda$. Here $nnz(\bm{Q})$ is the number of nonzero entries of $\bm{Q}$.

Additionally, there are a variety of sketching matrix constructions for $\bm{S}$ that can satisfy \eqref{eq:sketching_matrix2}. The running time of sketch $\bm{QS}$ depends on the dimension of $\bm{S}$, and the construction of $\bm{S}$ sampling with respect to $\emph{leverage scores}$ is proportional to $2K$ (we assume that $\rank{(\bm{Q})}=2K$), which means the running time of $\bm{QS}$ is also proportional to $2K$. Therefore, we let $\bm{S}$ dimensionality depend on the degrees of freedom $d_\lambda$ of the ridge regression problem, as opposed to the rank of matrix $\bm{Q}$. In this way, the running time would result in significant savings.

To achieve the reduction in running time, the column size $L$ of matrix $\bm{S}$ is thus better designed proportional to degrees of freedom $d_\lambda$ which depends on the distribution of the singular value of $\bm{Q}$ and $\lambda$ instead of proportional to $2K$ for $d_\lambda \leq 2K$, which could be significantly smaller than $2K$.

\begin{theorem}\label{thm:number_samples2}
Let $\bm{V}\in \mathbb{R}^{2M\times 2K}$ be the matrix of right singular vectors of $\bm{Q}$.
Let $\bm{S}$ be constructed by Algorithm \ref{Algorithm2} with the sampling probabilities $p_i = \| (\bm{V} \bm{\Sigma}_\lambda)_{i*} \|_2^2 /d_\lambda$ for $i=1,\ldots, 2M$. Let $\delta$ be a failure probability and let $0< \varepsilon \leq 1$ be an accuracy parameter. If the number of sampled columns $L$ satisfies
\begin{equation}\label{eq:set_equal}
L \geq \frac{8d_{\lambda}}{3 \varepsilon^2} \log{ \bigg( \frac{4(1+ d_{\lambda})}{\delta}    \bigg)},
\end{equation} 
then, with probability at least $1-\delta$,
\begin{equation}\label{eq:degree_freedom}
\| \bm{\Sigma}_\lambda \bm{V}^{\sf T} \bm{S} \bm{S}^{\sf T}  \bm{V} \bm{\Sigma}_\lambda - \bm{\Sigma}_{\lambda}^2 \|_2 \leq \varepsilon.
\end{equation}
\end{theorem}

The sampling probabilities $p_i = \| (\bm{V} \bm{\Sigma}_\lambda)_{i*} \|_2^2 /d_\lambda$ are the column \emph{ridge leverage scores} \cite{DBLP:conf/soda/CohenMM17,Alaoui:2015:FRK:2969239.2969326} of the channel matrix $\bm{Q}$. 
 Similarly to the constraint of ~\eqref{eq:sketching_matrix2}, setting $L=\mathcal{O} (d_\lambda \ln{d_\lambda})$ suffices to satisfy the condition \eqref{eq:degree_freedom}.
 Note that while the leverage scores which construct the sampling-and-rescaling matrix $\bm{S}$ with the column size $L$ depends on rank of $\bm{Q}$, the ridge leverage scores construct $\bm{S}$ depend on $d_\lambda$, which could be considerably smaller than rank of $\bm{Q}$. Hence, it could surely achieve time saving. However, the running time savings would lead to a drop in accuracy as shown in Theorem \ref{thm:main2}.

\subsection{Time Complexity}
We now discuss the time complexity of Algorithm \ref{main:Algorithm}.
Note that each column $\bm{M}$ in  \eqref{eq:rzf_sol} can be computed by each column of $\bm{\Lambda}$, separately. We consider $\bm{\Lambda}$ as a column vector.
Let $\bm{\Theta} =\bm{Q} \bm{S} \bm{S}^{\sf T}\bm{Q}^{\sf T} + \lambda \bm{I}_{2K} $.
Note that to find $\bm{\Theta}^{-1}$, it suffices to compute the singular value decomposition of $\bm{Q} \bm{S}$.
Since the singular values of $\bm{\Theta}$ can be computed through $\bm{\Sigma}_{\bm{Q} \bm{S}}+\lambda \bm{I}_{2K}$, where $\bm{\Sigma}_{\bm{A}}$ denotes the singular value of $\bm{A}$. And the left and right singular vectors of $\bm{\Theta}$ are the same as the left singular vectors of $\bm{Q} \bm{S}$. We store it implicitly by storing its left (and right) singular vector $\bm{U}_{\bm{\Theta}}$ and its singular values $\bm{\Sigma} _{\bm{\Theta}}$, before we just compute all the necessary matrix-vector products using this implicit representation of $\bm{\Theta}^{-1}$. The above analysis shows that we do not need to compute $\bm{\Theta}^{-1}$ directly. Thus computing $\bm{\Theta}^{-1}$ takes $\mathcal{O}(LK^2)$ time.
%
%
%
%
%
%

Updating each $\bm{\Lambda}^{(j)}$, $\bm{Y}^{(j)}$, and $\widetilde{\bm{M} }^{(j)}$ is dominated by the aforementioned running times, as all updates amount to just matrix-vector products.
Thus, summing over all $t$ iterations, the running time of Algorithm 
\ref{main:Algorithm} is given by
\begin{equation}
    \mathcal{O}(t\cdot nnz(\bm{Q})) 
      +\mathcal{O}(LK^2).
\end{equation}
 Thus the time complexity is reduced evidently. Note that the complexity of computing the matrix inversion  \eqref{eq:rzf_sol} is $\mathcal{O}(MK^2)$.

\section{The system sum-rate analysis with approximate RZF beamformers}
\label{sumrateprove}
In this section, we show that the system sum-rate of the randomized sketching based beamformer converges to the sum-rate of the RZF beamforming matrix as the number of iterations increases. Moreover, if an approximation sequence converges to the true beamforming matrix with the rate of convergence $\mathcal{O}(\beta_t)$, then the system sum-rate of the approximation sequence converges with the same rate of convergence $\mathcal{O}(\beta_t)$.
Before stating our main results, we introduce the extra notation, $\phi_{kj}$, to cast SINR at the $ k $-th user in  \eqref{eq:SINR} with a simpler form.
From now on, we assume that the channel matrix $\bm{H}$ is fixed and the beamforming matrix $\bm{W}$ is considered as complex variables.
Then we can easily deal with the system sum-rate for any approximate beamforming matrix.

For each $k,j$, let a function $\phi_{kj}:\C^{M \times K} \rightarrow [0,+\infty)$ be defined by
$\phi_{kj}(\bm{W}) = |\bm h_k^{\sf H}\bm w_j|^2$ for all $\bm{W}=   [\bm{w}_1, \cdots, \bm{w}_K]  \in \C^{M \times K} $.
Note that
\begin{align*}
\phi_{kj}(\bm{W})
&= \big( \re(\bm{h}_k)^{\sf T} \re(\bm{w}_j) - \im(\bm{h}_k)^{\sf T} \im(\bm{w}_j) \big)^2\\
&+ \big( \im(\bm{h}_k)^{\sf T} \re(\bm{w}_j) + \re(\bm{h}_k)^{\sf T} \im(\bm{w}_j) \big)^2 \geq 0.
\end{align*} 
The system sum-rate of  a given variables $\bm{W}$ can thus be rewritten  as
\begin{equation}\label{eq:RW_phi}
R(\bm{W}) = \sum_{k=1}^{K}\log\bigg(1+ \frac{ \phi_{kk}(\bm{W})}{\sum_{j\ne k} \phi_{kj}(\bm{W})+\sigma^2  }  \Bigg).
\end{equation}
That is, $R$ can be viewed as a function from $\C^{M \times K}$ to $[0,+\infty)$, as shown in Fig. \ref{fig:R}.

Let $V$ be a nonempty open subset of $\R^n$, $f:V \rightarrow \R^m$, and $p\in \N$.
Recall that a function $f$ is said to be $\mathcal{C}^p$ on $V$ if each partial derivative of $f$ of order $k \leq p$ exists and is continuous on $V$. 
$f$ is  said to be $\mathcal{C}^\infty$ on $V$ if $f$ is $\mathcal{C}^p$ on $V$ for all $p\in \N$.
In other words, a $\mathcal{C}^\infty$-mapping is a function that is differentiable for all degrees of differentiation.

\begin{figure}[!t]
        \centering
        \includegraphics[scale=0.6]{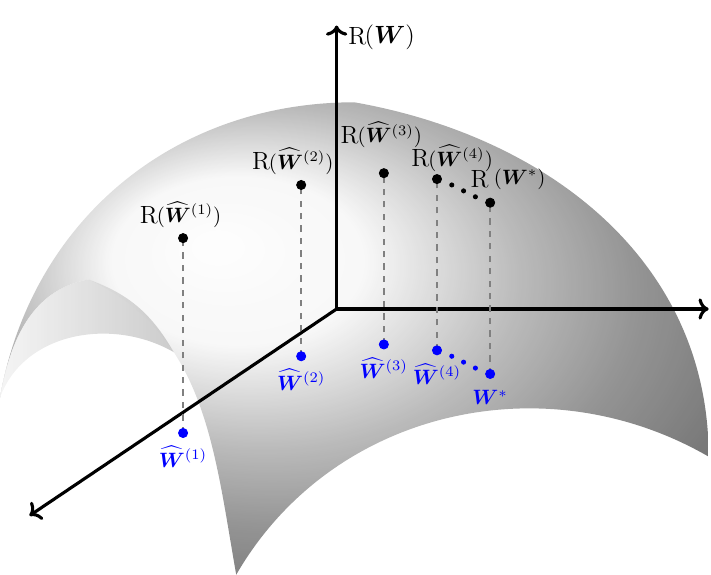}
        \caption{The system sum-rate $R(\bm{W})$.}
        \label{fig:R}
\end{figure}

\begin{lemma}\label{lemma:C-infinity}
The system sum-rate $R$  is a $\mathcal{C}^\infty$-mapping on $\R^{2M \times K}$.
\end{lemma}
\begin{proof}
Note that the complex variables $\bm{W} \in \C^{M \times K}$ can be considered as real variables $\bm{M} \in \R^{2M \times K}$.
We use $\bm{M}$ and $\bm{W}$ interchangeably.
Let $\bm{M} \in \R^{2M \times K}$ be real variables.
Then it is easy to check that 
 that $\phi_{kj}$ is a multivariate polynomial in $\R[\bm{M}]$, i.e., the ring of polynomials with real coefficients over variables $\bm{M}$.
Thus $\phi_{kj}$ is $C^\infty$-mapping on $\R^{2M \times 2K}$.
Since the logarithm function are $C^\infty$-mapping,
the function $R$ is a $C^\infty$-mapping on $\R^{2M \times 2K}$, provided
$\sum_{j\ne k} \phi_{kj}(\bm{W})+\sigma^2 \neq 0$.
Let $\bm{H} \in  \C^{K\times M} $ be a given channel matrix. Considering the beamforming matrix $\bm{M} \cong \bm{W}$  as real variables in $ \R^{2M \times K}( \cong  \C^{M \times K}) $, the system sum-rate in \eqref{eq:R} can be considered as a function $R:\R^{2M \times K} \longrightarrow [0,+\infty)$ defined by
\begin{equation*}
\bm{W}=[\bm{w}_1,\ldots,\bm{w}_K] \cong
\begin{bmatrix}
\re(\bm{w}_1) & \cdots & \re(\bm{w}_K) \\
\im(\bm{w}_1) & \cdots & \im(\bm{w}_K)   
\end{bmatrix}
 \longmapsto    R(\bm{W}).
\end{equation*}

In other words, $\phi_{kj}$ can be considered as a function from $\R^{2M \times 2K} \longrightarrow [0,+\infty)$.
Moreover, it is easy to check that $\phi_{kj}$ is a multivariate polynomial in $\R[\bm{M}]$, which is $C^\infty$-mapping on $\R^{2M \times 2K}$.
Since it can be rewritten  as
\begin{equation*}
R(\bm{W}) = \sum_{k=1}^{K}\log\bigg(1+ \frac{ \phi_{kk}(\bm{W})}{\sum_{j\ne k} \big(\phi_{kj}(\bm{W})+\sigma^2 \big) }  \Bigg),
\end{equation*}
and the logarithm function is $C^\infty$-mapping,
the function $R$ is a $C^\infty$-mapping on $\R^{2M \times 2K}$, provided
$\sum_{j\ne k} \big(\phi_{kj}(\bm{W})+\sigma^2 \big) \neq 0$.
\end{proof}

Denote the true solution of the regularized RZF problem \eqref{eq:rzf_sol} in the complex version as $\bm{W}^*$.
By Lemma \ref{lemma:C-infinity}, $R$ is continuous.
By Theorem \ref{thm:main1}, each entry of the approximation converges to the entry of the true solution, respectively, i.e.,
\begin{align*}
& \re( \widehat{\bm{W}}^{(t)}_{ij}) \longrightarrow \re(\bm{W}^*_{ij}) \quad \text{as} \quad t \longrightarrow \infty,\\
& \im( \widehat{\bm{W}}^{(t)}_{ij}) \longrightarrow \im(\bm{W}^*_{ij}) \quad \text{as} \quad t \longrightarrow \infty.
\end{align*}
Note that the image of a convergent sequence under a continuous function converges to the image of limit.
Thus, the following holds.
\begin{proposition}
Assume that for some constant $0 < \varepsilon <1$, the sketching matrix $\bm{S} \in \R^{2M\times L}$ satisfies the constraint \eqref{eq:contraint1}. 
Let $t$ be the number of iterations.
Then, the system sum-rate of the approximation $R(\widehat{\bm{W}}^{(t)})$
converges to the system sum-rate of the true solution $R(\bm{W}^*)$ as the number of iterations increases.
That is,
\begin{equation}
R\big( \widehat{\bm{W}}^{(t)} \big) \longrightarrow R\big( \bm{W}^* \big) \quad \text{as} \quad t \longrightarrow \infty.
\end{equation}
\end{proposition}

The next theorem is our key result. 
It shows that the error of the system sum-rate is bounded by the error of an approximation of beamforming matrix.
Using this result, the rate of convergence for the system sum-rate of an approximation can be obtained.
\begin{theorem}\label{thm:approx-R}
        Let $\bm{H}$ be a given channel matrix, and 
        let $\widehat{\bm{W}}$ (resp. $\bm{W}^*$) be the approximation (resp. true) RZF beamforming matrix.
        Then it holds that 
        \begin{align*}
        & \Big| R\big( \widehat{\bm{W}} \big) - R\big( \bm{W}^* \big)  \Big| \\
        & \leq   C \big\| \bm{H} \big\|_F^2 \Big( \big\|\widehat{\bm{W}} - \bm{W}^* \big\|_F^2     + 2\big\|\widehat{\bm{W}} - \bm{W}^* \big\|_F    \big\| \bm{W}^* \big\|_F  \Big),
        \end{align*}
        where $ C$ is constant independent to $\widehat{\bm{W}}$.
\end{theorem}
\begin{proof}
        See Appendix \ref{appendix_B}.
\end{proof}

Suppose a sequence $\{\beta_t\}_{t=1}^\infty$ converges to zero, and $\{ \alpha_t \}_{t=1}^\infty$ converges to a number $\alpha$. 
Recall that  $\{ \alpha_t \}_{t=1}^\infty$ converges  to $\alpha$ with {\emph{rate of convergence}} $\mathcal{O}(\beta_t)$ if
a positive constant $K$ exists with 
$ | \alpha_t  -\alpha | \leq K | \beta_t |$
for sufficiently large $t$.

The following shows that if an approximation sequence $\big\{\widehat{\bm{W}}^{(t)}  \big\}_{t=1}^{\infty}$ converges to the true beamforming matrix $\bm{W}^*$ with the rate of convergence  $\mathcal{O}(\beta_t)$, then $\big\{R( \widehat{\bm{W}}^{(t)} ) \big\}_{t=1}^{\infty}$ converges to  $R\big( \bm{W}^* \big)$ with  the same rate of convergence $\mathcal{O}(\beta_t)$.
\begin{theorem}\label{thm:general-convergence-rate-R}
Let $\eta \geq 0$. 
        If an approximation sequence $\{\widehat{\bm{W}}^{(t)} \}_{t=1}^{\infty}$ converges to $\bm{W}^*$ such that 
$       \|\widehat{\bm{W}}^{(t)} - \bm{W}^*\|_F \leq   |\beta_t|   (\| \bm{W}^*\|_F + \eta),$
        then 
        \begin{equation*}
        \Big| R\big( \widehat{\bm{W}}^{(t)} \big) - R\big( \bm{W}^* \big)  \Big|
        \leq  3C    |\beta_t  |     \big\| \bm{H} \big\|^2_F  (\big\| \bm{W}^{*} \big\|_F + \eta  )^2,
        \end{equation*}
        provided sufficiently large $t$. 
\end{theorem}
\begin{proof}
        Since $\beta_t$ converges to 0, there exists $T\in \N$ such that 
        $t\geq T$ implies $| \beta_t | <1$.
        By Theorem \ref{thm:approx-R}
        it follows that 
        \begin{align*}
        \big| R\big( \widehat{\bm{W}} \big) - R\big( \bm{W}^* \big)  \big|  
        &\leq C \big\| \bm{H} \big\|^2_F \bigg(  |\beta_t |^2 \big(\big\| \bm{W}^{*} \big\|_F +  \eta \big)^2 \\
        &\quad + 2  |\beta_t | \big(  \big\| \bm{W}^{*} \big\|_F + \eta \big) \big\|\bm{W}^* \big\|_F \bigg).
                \end{align*}
Since $\eta \geq 0$, we have    
        \begin{align*}
  &|\beta_t |^2 \big(\big\| \bm{W}^{*} \big\|_F + \eta \big)^2 + 2  |\beta_t | \big(  \big\| \bm{W}^{*} \big\|_F + \eta \big) \big\|\bm{W}^* \big\|_F  \\
        &\leq  |\beta_t |^2 \big(\big\| \bm{W}^{*} \big\|_F + \eta \big)^2 + 2  |\beta_t | \big(  \big\| \bm{W}^{*} \big\|_F + \eta \big)^2  \\
        &\leq  3|\beta_t | \big(\big\| \bm{W}^{*} \big\|_F + \eta \big)^2,
        \end{align*}
        provided $t \geq T$.
\end{proof}

Using Theorem \ref{thm:general-convergence-rate-R}, one can find the rate of convergence for the system sum-rate of the approximation sequence generated by Algorithm \ref{main:Algorithm}.
\begin{corollary}\label{cor:convergence-rate-R}
\begin{itemize}
    \item[(i)]  Assume that for $0 < \varepsilon <1$, the sketching matrix $\bm{S}$ satisfies the constraint \eqref{eq:contraint1}. Then, after $t$ number of iterations, the approximation $\widehat{\bm{W}}^{(t)}$ returned by Algorithm \ref{main:Algorithm} holds
        \begin{equation*}
        \Big| R\big( \widehat{\bm{W}}^{(t)} \big) - R\big( \bm{W}^* \big)  \Big|
        \leq  3C \varepsilon^{t} \big\| \bm{H} \big\|_F^2  \big\| \bm{W}^* \big\|_F^2.
        \end{equation*}
\item[(ii)] 
        Assume that for $0 < \varepsilon <1$, the sketching matrix $\bm{S}$ satisfies the constraint \eqref{eq:sketching_matrix4}. Then, after $t$ number of iterations, the approximation $\widehat{\bm{W}}^{(t)}$ returned by Algorithm \ref{main:Algorithm} holds
        \begin{align*}
        \Big| R\big( \widehat{\bm{W}}^{(t)} \big) - R\big( \bm{W}^* \big)  \Big| 
        &\leq  3C \varepsilon^{t} \big\| \bm{H} \big\|_F^2  \bigg( \big\| \bm{W}^* \big\|_F  \\
        &\quad + \frac{1}{\sqrt{2\lambda}} \| \bm{U}^{\sf T}_{2K,\perp}\bm{\Lambda} \|_2 \bigg)^2.
        \end{align*}
\end{itemize}
Here, $\bm{W}^*$ is the true value in \eqref{eq:bfd}.
\end{corollary}
\begin{proof}
(i) It holds from Theorem \ref{thm:general-convergence-rate-R} with $\eta=0$ and Theorem \ref{thm:main1}.
(ii) It holds from Theorem \ref{thm:general-convergence-rate-R} with $\eta=\frac{1}{\sqrt{2\lambda}} \| \bm{U}^{\sf T}_{2K,\perp}\bm{\Lambda} \|_2$ and Theorem \ref{thm:main1}.
\end{proof}

%
%
%
%

%
%

%
%
%

\section{Simulations}
\label{sec:simulations}

In this section, we simulate the performance of the proposed randomized sketching based beamformer in Algorithm \ref{main:Algorithm}. We consider the following channel model between the BS and the $ k $-th user:
        \begin{equation}
        \bm h_k=10^{-\tilde{L}(d_{k})/20}\sqrt{\varphi_ks_k}\bm f_k,
        \end{equation}
where $ \tilde{L}(d_k) $ is the path-loss at distance $ d_k $, $ s_k $ is the shadowing coefficients, $ \varphi_k $ is the antenna gain, and $ \bm f_k $ is the small fading coefficient. We use the standard cellular network parameters as shown in Table I \cite{shi2014group}.
We consider a single cell massive MIMO system with $ M=5000 $ antennas at the BS and $ K=50 $ single-antenna  users uniformly and independently distributed in the square region $ [-5000,5000]\times[5000,5000] $ meters. 
\begin{figure}[tp!]
\centering
\includegraphics[width=9cm,height=6.5cm]{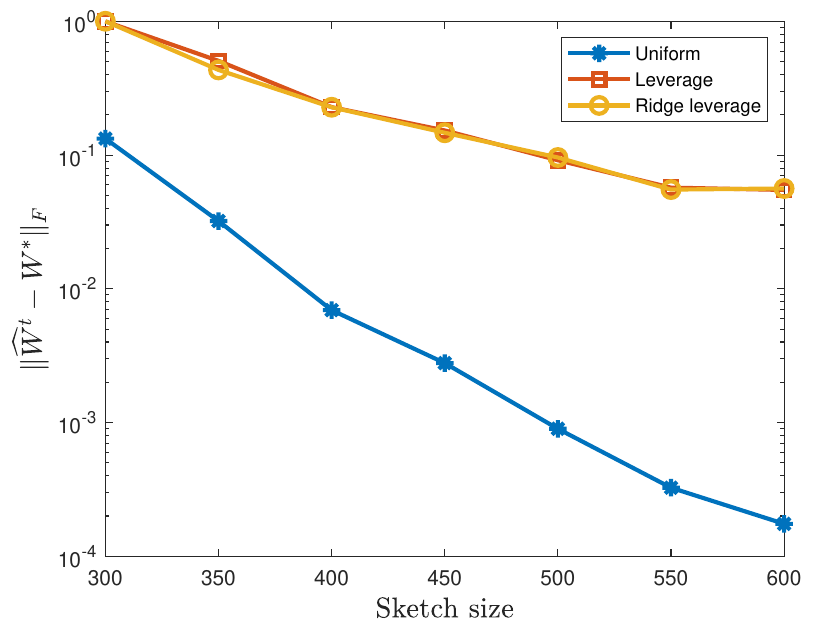}
\caption{Solution error vs. sketch size.}
\label{sim1}
\end{figure}

 First, we compare three different sampling-and-rescaling methods whose random matrices are generated by Algorithm \ref{Algorithm2} with the following sampling probabilities $\{p_i\}_{i=1}^{2M}$: 
        \begin{itemize}
                \item ({\bf Uniformly at random})
                Calculate 
                $p_i=\frac{1}{2M}$
                for $i=1,\dots, 2M$.
                \item ({\bf Leverage scores})
                Calculate $p_i=\frac{\|\bm{V}_{i*} \|^2_2}{2M}$ for $i=1,\dots,2M$. 
                \item ({\bf Ridge leverage scores}) 
                Calculate $p_i=\frac{\| (\bm{V} \bm{\Sigma}_\lambda)_{i*} \|^2_2}{d_\lambda}$ for $i=1,\dots, 2M$.
        \end{itemize}
        Here, $\bm{V}$ denotes the right singular value of matrix $\bm{Q}$, and 
        ${\bm{\Sigma}}_\lambda$ denotes the diagonal matrix given by \eqref{eq:sum}.
        
        Fig.~\ref{sim1} shows the sum-rate error for three different sampling-and-rescaling methods. We fix the iteration numbers $ t=10 $.
        Each graphs present the average of $200$  replicated runs.
        It clearly illustrates that the  uniformly  at random method achieves better accuracy than the other two sampling-and-rescaling methods. We thus  generate sampling matrix  uniformly  at random in the following simulations.

\begin{figure}[tp!]
        \begin{center}
                \includegraphics[width=9cm,height=6.5cm]{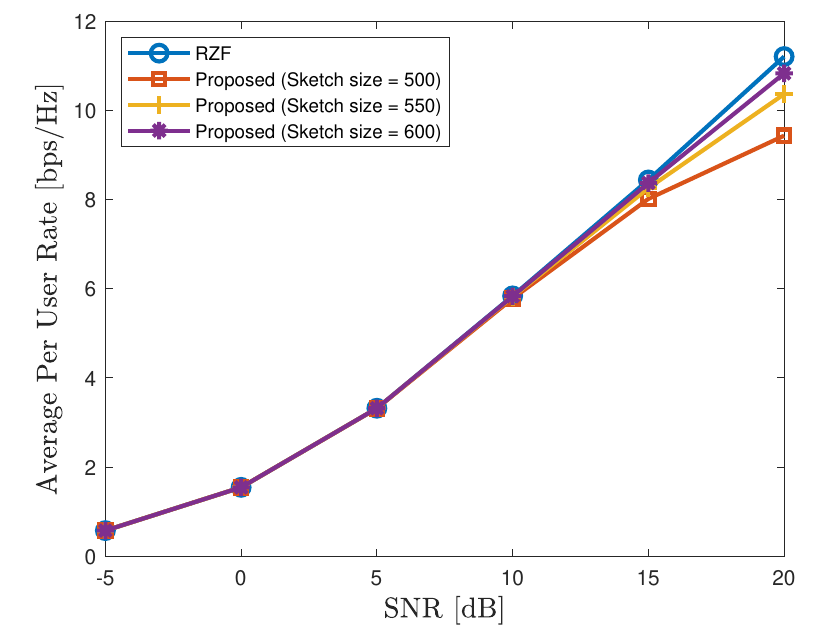}
                \caption{Average per user rate vs. {SNR}.}
                \label{sim5}
        \end{center}
\end{figure}

        We compare the proposed randomized sketching based beamformers under various sketch sizes and different SNR which is defined as the transmit power at the BS over the received noise power at all the users. 
        We generate sketching matrices with different sizes, and terminate Algorithm \ref{main:Algorithm} after $ 10 $ iterations. 
        As shown in Fig.~\ref{sim5}, the randomized sketching based beamformer performs closely to RZF beamforming in terms of the average per user rate as the sketch size increases.

\begin{figure}[tp!]
\begin{center}
\includegraphics[width=9cm,height=6.5cm]{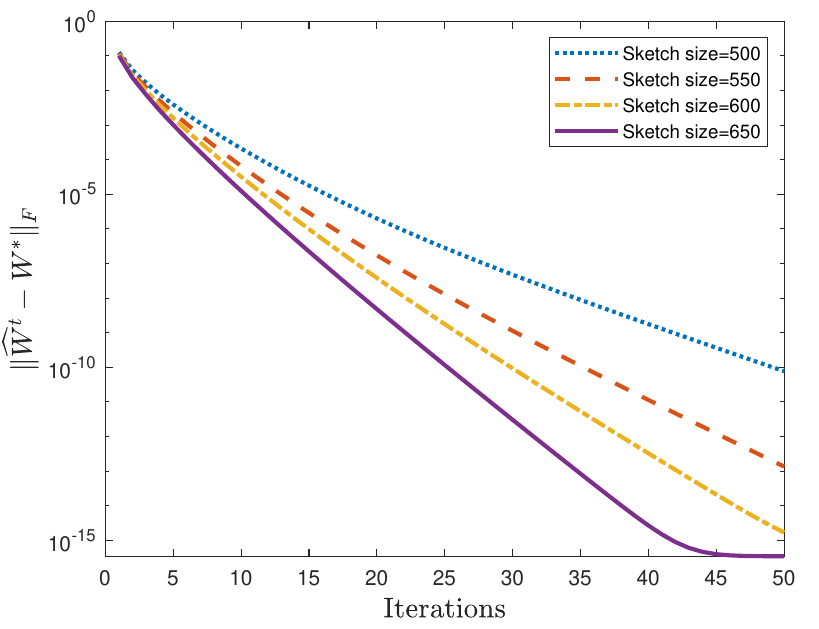}
\caption{Solution error vs. iteration.}
\label{sim2}
\end{center}
\end{figure}
\begin{figure}[tp!]
\begin{center}
\includegraphics[width=9cm,height=6.5cm]{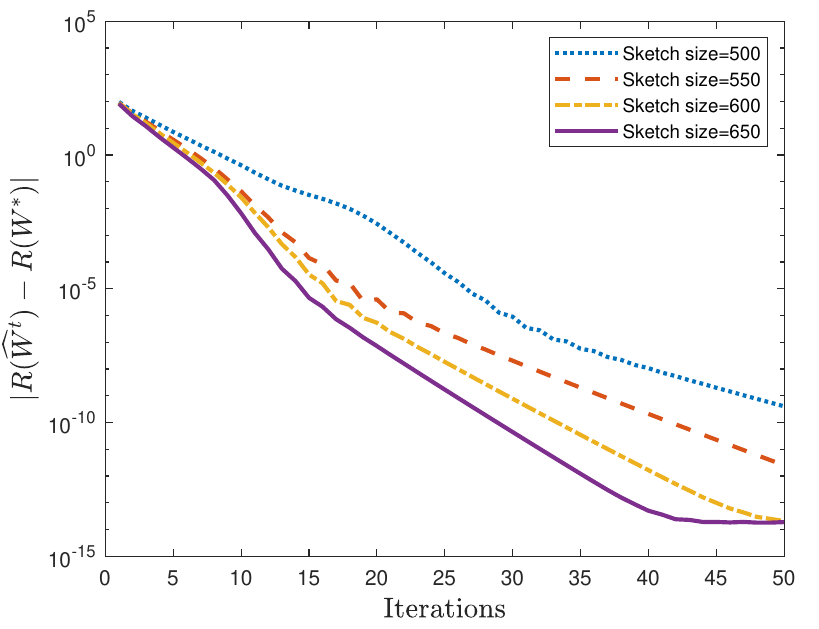}
\caption{Sum-rate error vs. iteration.}
\label{sim3}
\end{center}
\end{figure}

Fig.~\ref{sim2} illustrates that the iterative solution converges to the RZF beamforming matrix at a linear convergence rate as shown in Theorem \ref{thm:main1} by plotting the trend of $\|\widehat{\bm{W}}^{(t)} - \bm{W}^*\|_F$ up to 50 iterations with SNR being $ 5 $. Fig.~\ref{sim3} illustrates that the achievable sum-rate of the randomized beamforming converges to the achievable sum-rate given by RZF beamforming linearly as shown in Corollary \ref{cor:convergence-rate-R}. It demonstrates the trend of $| R\big( \widehat{\bm{W}}^{(t)} \big) - R\big( \bm{W}^* \big)|$ within 50 iterations. It is clear that in Fig.~\ref{sim2} and \ref{sim3} the error decreases fast with the number of iterations, and larger sketch size leads to faster convergence rate.

\begin{figure}[t!]
\begin{center}
\includegraphics[width=9cm,height=6.5cm]{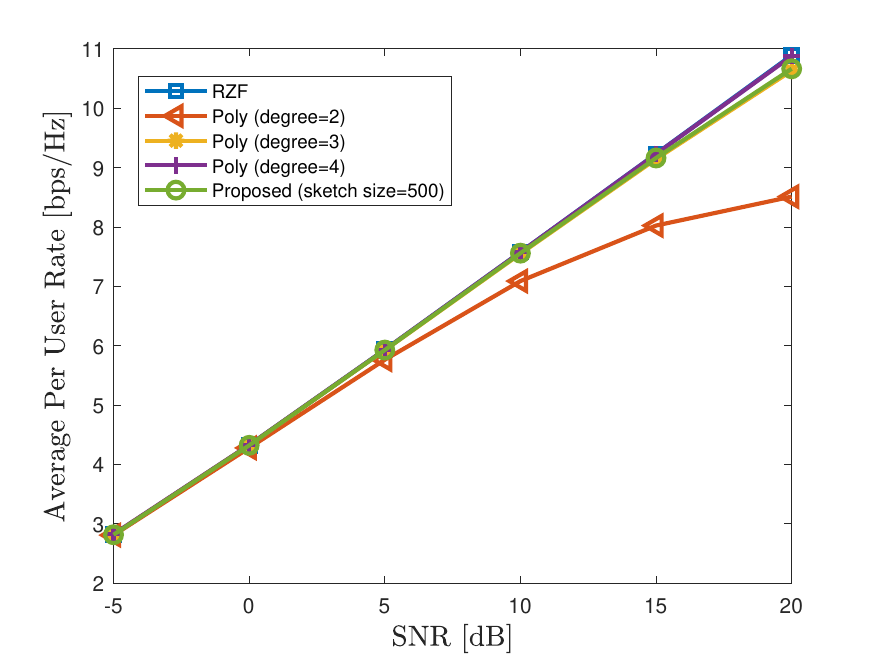}
\caption{Compassion with truncated polynomial expansion method\cite{mueller2016linear}.}
\label{sim6}
\end{center}
\end{figure}

Finally, we compare the proposed randomized sketching method with  the polynomial expansion based method \cite{mueller2016linear} given the channel matrix $\bm H$. Note that the polynomial expansion based method \cite{mueller2016linear} consists of two steps, i.e., seeking the polynomial coefficients based on the channel matrix $\bm H$ and then computing the beamforming matrix using the polynomial coefficients. We follow the simulation settings in  \cite{mueller2016linear} and consider the case $ M=1000$, $K=50$. The sketch size is set to be $ 500 $, and we terminate our proposed sketch method after $ 10 $ iterations. The results averaged for $ 100 $ times are presented in Fig. \ref{sim6}.  For the polynomial expansion based method, we present the results with polynomial expansion degree equals $ 2,3 $ and $ 4 $, respectively.  As can be seen from Fig. \ref{sim6}, our proposed method and  the truncated polynomial expansion method with degree $3$ achieve similar performance as the RZF method. To compute the beamforming matrix, the average running time of our proposed randomized sketching method is $0.0214 s$, while it costs $2.0950s$ via the truncated polynomial expansion method with degree $3$ (see TABLE \ref{table:timecost}). This is because the time complexity is high when we compute the polynomial coefficients. Therefore, given the channel matrix $\bm H$, the proposed sketching method is more efficient compared to the polynomial expansion based method.  

\begin{table}[htb]
        \centering
        \caption{Comparison for time complexity}
        \begin{tabular}{ll}
        \toprule
         Approaches &Time (seconds)\\
         \midrule
         Proposed (sketch size = 500) &0.0214\\
         Poly (degree = 2)  &0.6728\\
         Poly (degree = 3)  &2.0950\\
         Poly (degree = 4)  &5.2871\\  
        \bottomrule
        \end{tabular}
        \label{table:timecost}        
\end{table}
        

%
\section{Conclusion}
\label{sec:conclusion}
We proposed a randomized sketching based RZF beamforming approach to tackle the computational challenges of precoding in massive MIMO systems. This  was achieved by solving the linear system for the matrix inversion via  randomized sketching based on the preconditioned Richard iteration. The computational complexity of our proposed method scales with $ LK^2 $, where $ L\ll 2M $ is the sketching matrix size.  Furthermore, we proved that the proposed algorithm iteratively converges to the RZF beamforming matrix at a linear convergence rate. Also, the achievable sum-rate with the randomized sketching based RZF beamformer linearly converges to the achievable sum-rate with the RZF beamformer  as the number of iteration increases. Simulation results were demonstrated to verify our theoretical findings.

\appendices

\section{Minor results}
\begin{lemma}\label{lem:log}
        For all $a,b>0$, it holds that 
        $| \log(1+a) -\log(1+b) | \leq |a-b |.$
\end{lemma}
\begin{proof}
Let $f(x)= \log(1+x)$ and $a,b>0$ with $a\neq b$.
By the Mean Value Theorem, there exists $\xi \in (a,b)$ such that
\begin{equation*}
\bigg| \frac{f(a)-f(b)}{a-b} \bigg| = | f'(\xi) | <1.
\end{equation*}
\end{proof}

\begin{lemma}\label{lem:inequality1}
        Let $a,b>0$ be given. Then it holds that 
        \begin{equation}
        \Big|    \frac{y}{x+b} - \frac{a}{b}  \Big|  \leq   \frac{1}{|b|} |y-a| + \frac{a}{b^2} |x|,
        \end{equation} 
        for all $x,y\geq 0$.
\end{lemma}
\begin{proof}
By the triangle inequality, it holds that 
\begin{align*}
\Big|    \frac{y}{x+b} - \frac{a}{b}  \Big|  
=     \frac{| b(y-a)-ax| }{|x+b| |b|}  
 \leq   \frac{| b(y-a)| + |ax| }{|b|^2}
\end{align*} 
for all $x,y \geq 0$.
\end{proof}

\begin{lemma}\label{lem:phi1}
        For each $k,j$ it holds that
        \begin{align*}
        & \Big| \phi_{kj}(\widehat{\bm{W}}) -\phi_{kj}(\widetilde{\bm{W}}) \Big| \\
        & \leq \big\| \widehat{\bm{w}}_j - \widetilde{\bm{w}}_j  \big\|_2   \big\| \bm{h}_k \big\|_2^2   \big(  \big\| \widehat{\bm{w}}_j -  \widetilde{\bm{w}}_j  \big\|_2 + 2\big\| \widetilde{\bm{w}}_j \big\|_2   \big).
        \end{align*}
\begin{proof}
                Using the fact $ |\bm{h}_k^{\sf H} \bm{w}_j|^2 =  \bm{w}_j^{\sf H} \bm{h}_k \bm{h}_k^{\sf H} \bm{w}_j$,
                by the triangle inequality, we have that
                \begin{align*}
                & \Big| \phi_{kj}(\widehat{\bm{W}}) -\phi_{kj}(\widetilde{\bm{W}}) \Big| \\
                &=  \big|   \widehat{\bm{w}}^{\sf H}_j \bm{h}_k\bm{h}_k^{\sf H} \widehat{\bm{w}}_j    -     \widetilde{\bm{w}}^{\sf H}_j \bm{h}_k\bm{h}_k^{\sf H}  \widetilde{\bm{w}}_j   \big| \\
                &=  \big|   \widehat{\bm{w}}^{\sf H}_j \bm{h}_k\bm{h}_k^{\sf H} (\widehat{\bm{w}}_j - \widetilde{\bm{w}}_j  )
                + (\widehat{\bm{w}}^{\sf H}_j -  \widetilde{\bm{w}}^{\sf H}_j)  \bm{h}_k\bm{h}_k^{\sf H}  \widetilde{\bm{w}}_j   \big| \\
                &\leq  \big|   \widehat{\bm{w}}^{\sf H}_j \bm{h}_k\bm{h}_k^{\sf H} (\widehat{\bm{w}}_j - \widetilde{\bm{w}}_j  )  \big| 
                +  \big|  (\widehat{\bm{w}}^{\sf H}_j -  \widetilde{\bm{w}}^{\sf H}_j)  \bm{h}_k\bm{h}_k^{\sf H}   \widetilde{\bm{w}}_j \big| \\
                & \leq  
                \big\| \widehat{\bm{w}}_j -  \widetilde{\bm{w}}_j  \big\|_2   \big\|  \bm{h}_k\bm{h}_k^{\sf H}  \big\|_2   \big( \big\|\widehat{\bm{w}}_j \big\|_2 + \big\|  \widetilde{\bm{w}}_j  \big\|_2  \big). \\
                \end{align*}
By the triangle inequality and the definition of operator norm,
it holds that
\begin{align*}
& \big\| \widehat{\bm{w}}_j - \bm{w}_j  \big\|_2   \big\|  \bm{h}_k\bm{h}_k^{\sf H}  \big\|_2   \big( \big\|\widehat{\bm{w}}_j \big\|_2 + \big\| \bm{w}_j  \big\|_2  \Big) \\
& \leq \big\| \widehat{\bm{w}}_j - \bm{w}_j  \big\|_2   \big\|  \bm{h}_k\bm{h}_k^{\sf H}  \big\|_2   \big(  \big\| \widehat{\bm{w}}_j - \bm{w}_j   \big\|_2 + 2\big\| \bm{w}_j \big\|_2   \big) \\
&  \leq \big\| \widehat{\bm{w}}_j - \bm{w}_j  \big\|_2   \big\| \bm{h}_k \big\|^2_2   \big(  \big\| \widehat{\bm{w}}_j -  \bm{w}_j  \big\|_2 + 2\big\| \bm{w}_j \big\|_2   \big).
        \end{align*}
                \end{proof}
\end{lemma}


\section{PROOF OF THEOREM \ref{thm:approx-R}}\label{appendix_B}

%


Since ${\sf{SINR}}_k>0$ for all $k$, the triangle inequality and Lemma \ref{lem:log}
imply that
$$\big| R\big( \widehat{\bm{W}} \big) - R\big( \bm{W}^* \big)  \big|
\leq  \sum_{k=1}^{K} \Big| {\sf{SINR}}_k( \widehat{\bm{W}}  ) - {\sf{SINR}}_k(\bm{W}^*)  \Big|.$$
Then by Lemma \ref{lem:inequality1} it follows that 
\begin{align*}
& \sum_{k=1}^{K} \Big| {\sf{SINR}}_k( \widehat{\bm{W}}  ) - {\sf{SINR}}_k(\bm{W}^*)  \Big|\\
&=\sum_{k=1}^{K}   \Bigg| \frac{ \phi_{kk}(\widehat{\bm{W}})}{\sum_{j\ne k}  \phi_{kj}(\widehat{\bm{W}})+\sigma^2 } - \frac{ \phi_{kk}(\bm{W}^*)}{\sum_{j\ne k}  \phi_{kj}(\bm{W}^*)+\sigma^2 }  \Bigg| \\
& \leq  \sum_{k=1}^{K} \Bigg[  \frac{   1       }{\sum_{j\ne k}  \phi_{kj}(\bm{W}^*)+\sigma^2 }   \Big| \phi_{kk}(\widehat{\bm{W}}) -\phi_{kk}(\bm{W}^*) \Big|        \\                                      
& \quad + \frac{\phi_{kk}(\bm{W}^*)}{(\sum_{j\ne k}  \phi_{kj}(\bm{W}^*)+\sigma^2)^2 }  \bigg|  \sum_{j\ne k}  \phi_{kj}(\widehat{\bm{W}}) - \sum_{j\ne k}  \phi_{kj}(\bm{W}^*) \bigg| \Bigg] \\
& \leq \frac{C}{2} \sum_{k=1}^{K} \bigg[   \big| \phi_{kk}(\widehat{\bm{W}}) -\phi_{kk}(\bm{W}^*) \big|                                              
+ \sum_{j\ne k}   \big|  \phi_{kj}(\widehat{\bm{W}}) -  \phi_{kj}(\bm{W}^*) \big| \bigg] \\
& \leq   C \|  \bm{H} \|_F^2   \sum_{k=1}^{K} \Big( \big\| \widehat{\bm{w}}_k - \bm{w}^*_{k}  \big\|_2^2     +2  \big\|  \widehat{\bm{w}}_k -  \bm{w}_k^* \big\|_2 \big\| \bm{w}_k^* \big\|_2   \Big) \\
& \leq     C  \| \bm{H} \|_F^2     \Big( \big\|\widehat{\bm{W}} - \bm{W}^* \big\|_F^2     + 2\big\|\widehat{\bm{W}} - \bm{W}^* \big\|_F    \big\| \bm{W}^* \big\|_F  \Big),
\end{align*}
where  
$
C =  2\underset{k}{\max}\Bigg\{  \frac{         1       }{\sum_{j\ne k}  \phi_{kj}(\bm{W}^*)+\sigma^2 },  \frac{\phi_{kk}(\bm{W}^*)}{(\sum_{j\ne k}  \phi_{kj}(\bm{W}^*)+\sigma^2)^2 } \Bigg\}.
$
Note that the first inequality holds from Lemma \ref{lem:inequality1}.
The last second inequality holds from the following by Lemma \ref{lem:phi1}. 
{\small
\begin{align*}
& \sum_{k=1}^{K} \sum_{j\ne k}  \big|  \phi_{kj}(\widehat{\bm{W}}) - \phi_{kj}(\bm{W}^*) \big| \\
&\qquad  \leq  \sum_{j=1}^K   \sum_{k=1}^{K}  \big\| \widehat{\bm{w}}_j - \bm{w}_j^*  \big\|_2   \big\| \bm{h}_k \big\|_2^2   \big(  \big\| \widehat{\bm{w}}_j -  \bm{w}_j^*  \big\|_2 + 2\big\| \bm{w}_j^* \big\|_2   \big)\\ 
&\qquad  \leq   \big\| \bm{H} \big\|_F^2 \sum_{j=1}^K  \Big( \big\| \widehat{\bm{w}}_j - \bm{w}^*_{j}  \big\|_2^2     +2  \big\|  \widehat{\bm{w}}_j -  \bm{w}_j^* \big\|_2 \big\| \bm{w}_j^* \big\|_2   \Big), \\
& \sum_{k=1}^{K} \Big| \phi_{kk}(\widehat{\bm{W}}) -\phi_{kk}(\bm{W}^*) \Big|  \\
& \qquad \leq   \sum_{k=1}^{K} \big\| \widehat{\bm{w}}_k - \bm{w}_k^*  \big\|_2   \big\| \bm{h}_k \big\|_2^2   \big(  \big\| \widehat{\bm{w}}_k -  \bm{w}_k^*  \big\|_2 + 2\big\| \bm{w}_k^* \big\|_2   \big)\\
& \qquad  \leq   \sum_{k=1}^{K}  \big\| \widehat{\bm{w}}_k - \bm{w}^*_{k}  \big\|_2   \big\| \bm{H} \big\|_F^2   \big(   \big\|  \widehat{\bm{w}}_k -  \bm{w}_k^*  \big\|_2 + 2\big\| \bm{w}_k^* \big\|_2   \big)\\
& \qquad \leq  \|  \bm{H} \|_F^2 \sum_{k=1}^{K} \Big( \big\| \widehat{\bm{w}}_k - \bm{w}^*_{k}  \big\|_2^2     +2  \big\|  \widehat{\bm{w}}_k -  \bm{w}_k^* \big\|_2 \big\| \bm{w}_k^* \big\|_2   \Big).
\end{align*}
}
The last inequality holds from the Cauchy-Schwartz inequality as follows:
\begin{align*}
&\sum_{k=1}^{K}  \big\|  \widehat{\bm{w}}_k -  \bm{w}_k^* \big\|_2 \big\| \bm{w}_k^* \big\|_2  \\
 &\leq \Bigg( \sum_{k=1}^{K}  \big\|  \widehat{\bm{w}}_k -  \bm{w}_k^* \big\|_2^2  \Bigg)^{\frac{1}{2}}   
\Bigg( \sum_{k=1}^{K}  \big\| \bm{w}_k^* \big\|_2^2   \Bigg)^{\frac{1}{2}}\\
& \leq   \big\|\widehat{\bm{W}} - \bm{W}^* \big\|_F    \big\| \bm{W}^* \big\|_F.
\end{align*}

\vspace{1cm}

\bibliographystyle{ieeetr}
\bibliography{refs}

\end{document}